\documentclass[12pt,a4paper]{article}
\usepackage{amsmath} 
\usepackage{amsthm} 
\usepackage[english]{babel}
\usepackage{epsfig} 
\usepackage{float} 
\usepackage{natbib} 
\usepackage{apalike} 
\usepackage{bbold} 
\usepackage{color}  

\usepackage{subfig} 
\usepackage{bbold} 
\usepackage{dsfont} 
\usepackage[ruled,norelsize]{algorithm2e} 

\newcommand{\suchthat}{\;\ifnum\currentgrouptype=16 \middle\fi|\;}

\newcommand{\cov}{\mathrm{Cov}}

\newtheorem{prop}{Proposition}

\newtheorem{defi}{Definition}
\newtheorem{rem}{Remark}

\title{Finite-dimensional approximation of Gaussian processes with inequality constraints}
\begin{document}
\maketitle
\begin{center}
\textbf{H. Maatouk}\\
\vskip0.2cm
INRIA Centre de Recherche Rennes - Bretagne Atlantique\\
Campus de Beaulieu, 35042 Rennes, France \\
hassan.maatouk@inria.fr
\end{center}
\vspace{2cm}

\textbf{Abstract}
Due to their flexibility, Gaussian processes (GPs) have been widely used in nonparametric function estimation.
A prior information about the underlying function is often available. For instance, the physical system (computer model output) may be known to satisfy inequality constraints with respect to some or all inputs. We develop a finite-dimensional approximation of GPs capable of incorporating inequality constraints and noisy observations for computer model emulators. It is based on a linear combination between Gaussian random coefficients and deterministic basis functions. By this methodology, the inequality constraints are respected in the entire domain. The mean and the maximum of the posterior distribution are well defined. A simulation study to show the efficiency and the performance of the proposed model in term of predictive accuracy and uncertainty quantification is included.

\bigskip


\section{Introduction and related work}\label{intro}
In the estimation of nonparametric function, Gaussian processes (GPs) are the most popular choices. This is because of their flexibility and other nice properties. For instance, the conditional GP with linear equality constraints is still a GP \citep{cramer1967stationary}. Additionally, some inequality constraints (such as monotonicity and convexity) of output computer responses are related to partial derivatives. The partial derivatives of the GP remain GPs \citep{cramer1967stationary,opac-b1081425}. Incorporating an infinite number of linear inequality constraints (such as boundedness, monotonicity and convexity) into a GP model is a difficult problem. This is because the resulting conditional process is not a GP in general.\

Constrained GPs (or kriging) has been studied in the domain of geostatistics \citep{Freulon1993,kleijnen2012monotonicity}. In the literature, there are a variety of ways for incorporating linear inequality constraints into a GP emulator. In \citet{Abrahamsen2001}, the idea is based on a discrete location approximation. In that case, the inequality constraints are satisfied in a finite number of input locations. For monotonicity and isotonicity constraints, some methodologies are based on the knowledge of the derivatives of the GP at some input locations \citep{golchi2015monotone,journals/jmlr/RiihimakiV10,wang2016estimating}. As mentioned in \citet{wang2016estimating}, `only a modest number of virtual derivative points seems to be needed to effectively impose the desired shape constraint'. In \citet{lin2014bayesian}, Gaussian process projection is studied. A comparison with spline-based models is included. Recently, a new methodology based on a modification of the covariance function in Gaussian processes to correctly account for known linear constraints is developed in \citet{jidling2017linearly}. \

For monotone function estimations, using B-splines was firstly introduced by \citet{ramsay1988,ramsay1998estimating}. The idea is based on the integration of B-splines defined on a properly set of knots with positive coefficients to ensure monotonicity constraints. \citet{He96monotoneb-spline} take the same approach and suggest the calculation of the coefficients by solving a finite linear minimization problem. In \citet{delecroix1996functional}, nonparametric function estimation in a general cone is studied. Their method is based on a projection into a discretized version of the cone, using the theory of reproducing kernel Hilbert spaces. In \citet{shively2009bayesian}, a Bayesian approach to estimate nonparametric monotone functions using restricted splines is developed. In \citet{saarela2011method,doi:10.1198/106186007X180949}, the generalization of monotonic regression to multiple dimensions are studied.\

The methodology developed in the present paper is quite different. It is based on a finite-dimensional approximation of GPs (or a GP approximation) that converges uniformly pathwise. It can be seen as a linear combination between deterministic basis functions and Gaussian random coefficients, where the coefficients are not independent. The main idea is to choose the basis  functions such that the infinite number of inequality constraints on the GP approximation are equivalent to a finite number of constraints on the coefficients. Therefore, the simulation of the conditional GP approximation is reduced to the simulation of a Gaussian vector (random coefficients) restricted to convex sets which is a well-known problem with existing algorithms \citep{Botts,Chopin2011FST19607241960748,Maatouk2016,journals/sac/PhilippeR03,Robert}.\

The article is structured as follows. In Section~\ref{GPCE}, Gaussian processes for computer experiments, their derivative processes and the choice of covariance functions are briefly reviewed. In Section~\ref{FDAGPs}, a finite-dimensional approximation of GPs capable of incorporating inequality constraints and noisy observations is developed. Section~\ref{IllExamples} shows some simulated examples of the finite-dimensional approximation of GPs conditionally to inequality constraints (such as boundedness and monotonicity) and noisy observations in one and two dimensions. Section~\ref{SimStudy} investigates the performance of the proposed model in terms of predictive accuracy and uncertainty quantification.

\section{Gaussian processes for computer experiments}\label{GPCE}
The following model is considered 
\begin{equation*}
y=f(\boldsymbol{x}), \quad \boldsymbol{x}\in\mathbb{R}^d,
\end{equation*}
where the simulator response $y$ is assumed to be a deterministic real-valued function of the $d$-dimensional variable $\boldsymbol{x}=(x_1,\ldots,x_d)\in \mathbb{R}^d$. The true function is supposed to be continuous and evaluated at data of size $n$ (design of experiments) given by the rows of the $n\times d$ matrix $\boldsymbol{X}=\left(\boldsymbol{x}^{(1)},\ldots,\boldsymbol{x}^{(n)}\right)^\top$, where $\boldsymbol{x}^{(i)}\in \mathbb{R}^d, \ i=1\ldots,n$. In many practical situations, it is not possible to get exact evaluations of $y$ at the design of experiments, but rather pointwise noisy measurements. In such case, an approximate response $y(\boldsymbol{X})+\epsilon$ is available, where $\epsilon\sim\mathcal{N}(\boldsymbol{0},\sigma_{\text{noise}}^2\textbf{I})$ with $\sigma_{\text{noise}}^2$ the noise variance and $\textbf{I}$ the identity matrix. To simplify notations, we denote $\tilde{y}_i=y(\boldsymbol{x}^{(i)})+\epsilon_i$, $i=1,\ldots,n$. In the statistical framework, $y$ is viewed as a realization of a continuous GP
\begin{equation*}
Z(\boldsymbol{x})=\eta(\boldsymbol{x})+Y(\boldsymbol{x}), \quad \boldsymbol{x}\in\mathcal{D}\subset\mathbb{R}^d,
\end{equation*}
where $\mathcal{D}$ is a compact subset of $\mathbb{R}^d$ and the deterministic continuous function $\eta~: \ \boldsymbol{x}\in \mathbb{R}^d  \ \longrightarrow \ \eta (\boldsymbol{x})\in \mathbb{R}$ is the mean and $Y$ is a zero-mean GP with continuous covariance function 
\begin{equation*}
K~: \ (\boldsymbol{x},\boldsymbol{x}')\in\mathcal{D}\times\mathcal{D} \  \longrightarrow \ K(\boldsymbol{x},\boldsymbol{x}')\in\mathbb{R}.
\end{equation*} 
In that case, the GP can be written as $Z\sim \mathcal{GP}\left(\eta(\boldsymbol{x}),K(\boldsymbol{x},\boldsymbol{x}')\right)$. Conditionally to noisy observations $\boldsymbol{\tilde{y}}=(\tilde{y}_1,\ldots,\tilde{y}_n)^\top$, the process remains a GP
\begin{equation*}
Z(\boldsymbol{x}) \suchthat Z\left(\boldsymbol{X}\right)=\boldsymbol{\tilde{y}}\sim\mathcal{GP}
\left(\zeta(\boldsymbol{x}),\tau^2(\boldsymbol{x})\right),
\end{equation*}
where
\begin{eqnarray}\label{condGP}
&&\zeta(\boldsymbol{x})=\eta(\boldsymbol{x})+\boldsymbol{k}(\boldsymbol{x})^\top (\mathbb{K}+\sigma_{\text{noise}}^2\textbf{I})^{-1}\left(\boldsymbol{\tilde{y}}-\boldsymbol{\mu}\right); \\ \nonumber
&&\tau^2(\boldsymbol{x})=K(\boldsymbol{x},\boldsymbol{x})-\boldsymbol{k}(\boldsymbol{x})^\top (\mathbb{K}+\sigma_{\text{noise}}^2\textbf{I})^{-1}\boldsymbol{k}(\boldsymbol{x}),
\end{eqnarray}
and 
$\boldsymbol{\mu}=\eta(\boldsymbol{X})$ is the vector of trend values at the design of experiments,
$\mathbb{K}_{i,j}=K(\boldsymbol{x}^{(i)},\boldsymbol{x}^{(j)})$, $i,j=1,\ldots,n$ is the covariance matrix of $Z(\boldsymbol{X})$ and
$\boldsymbol{k}(\boldsymbol{x})=K(\boldsymbol{x}, \boldsymbol{X})$ is the vector of covariance between $Z\left(\boldsymbol{x}\right)$ and $Z\left(\boldsymbol{X}\right)$. Additionally, the covariance function between any two inputs is equal to
\begin{equation*}
C(\boldsymbol{x},\boldsymbol{x}')=\cov(Z(\boldsymbol{x}),Z(\boldsymbol{x}') \suchthat Z(\boldsymbol{X})=\boldsymbol{\tilde{y}})=K(\boldsymbol{x},\boldsymbol{x}')-\boldsymbol{k}(\boldsymbol{x})^\top(\mathbb{K}+\sigma_{\text{noise}}^2\textbf{I})^{-1}\boldsymbol{k}(\boldsymbol{x}'),
\end{equation*}
where $C$ is the covariance function of the conditional GP. The mean $\zeta(\boldsymbol{x})$ is called kriging mean prediction of $Z(\boldsymbol{x})$ based on the computer model outputs $Z\left(\boldsymbol{X}\right)=\boldsymbol{\tilde{y}}$ \citep{Rasmussen2005GPM1162254}.

\subsection{The choice of covariance function}
The covariance function $K$ controls the smoothness of the kriging metamodel. It must be chosen in the set of definite and positive kernels. In Table~\ref{kernels}, some popular covariance functions used in kriging methods are given. Notice that these covariance functions are placed in decreasing order of smoothness, the squared exponential covariance function corresponding to $\mathcal{C}^{\infty}$ function (i.e., the space of functions that admit derivatives of all orders) and the exponential covariance function to continuous one \citep{Rasmussen2005GPM1162254}.

\begin{table*}
\centering
\caption{Some popular covariance functions used in kriging methods.}
\begin{tabular}{*{3}{c}}
Name  &  Expression &  Class   \\   \hline
Squared exponential       & $\sigma^2\exp\left(-\frac{(x-x')^2}{2\theta^2}\right)$ 	& $\mathcal{C}^{\infty}$       \\
Mat\'ern 5/2   & $\sigma^2\left(1+\frac{\sqrt{5}\mid x-x'\mid}{\theta}+\frac{5(x-x')^2}{3\theta^2}\right)\exp\left(-\frac{\sqrt{5} | x-x'|}{\theta}\right)$
& $\mathcal{C}^2$      \\
Mat\'ern 3/2   &  $\sigma^2\left(1+\frac{\sqrt{3}\mid x-x'\mid}{\theta}\right)\exp\left(-\frac{\sqrt{3}|x-x'|}{\theta}\right)$   & $\mathcal{C}^1$      \\
Exponential    &  $\sigma^2\exp\left(-\frac{|x-x'|}{\theta}\right)$ & $\mathcal{C}^0$      \\  \hline 
\end{tabular}
\label{kernels}
\end{table*}

\subsection{Derivatives of Gaussian processes}
In this subsection, the paths of the GP $(Z(\boldsymbol{x}))_{\boldsymbol{x}\in\mathbb{R}^d}$ are assumed to be of class $\mathcal{C}^p$ (i.e., the space of functions that admit derivatives up to order $p$). This can be guaranteed if $K$ is smooth enough, and in particular if $K$ is of class $\mathcal{C}^{\infty}$ \citep{cramer1967stationary}. 
Since differentiation is a linear operator, the order partial derivatives of a GP remain GPs \citep{cramer1967stationary,opac-b1081425} and
\begin{eqnarray*}
E\left(\partial_{x_k}^pZ(\boldsymbol{x})\right)&=&\frac{\partial^p}{\partial x_k^p}\eta(\boldsymbol{x}),\\
\cov\left(\partial_{x_k}^pZ(\boldsymbol{x}^{(i)}),\partial_{x_{\ell}}^q Z(\boldsymbol{x}^{(j)})\right)&=&\frac{\partial^{p+q}}{\partial x_k^p\partial (x'_{\ell})^q}K(\boldsymbol{x}^{(i)},\boldsymbol{x}^{(j)}).
\end{eqnarray*}

\section{Finite-dimensional approximation of GPs}\label{FDAGPs}
Without loss of generality the input set is supposed the unit hypercube $\mathcal{D}=[0,1]^d$. The input set $\mathcal{D}$ is discretized uniformly to $(N+1)^d$ knots. For example, in one dimension where $\mathcal{D}=[0,1]$, the discretization can be summarized as follow: $0=t_{N,0},\ldots,t_{N,N}=1$. Let $Y\sim\mathcal{GP}(0,K(\boldsymbol{x},\boldsymbol{x}'))$ be a zero-mean GP with covariance function $K$. The finite-dimensional approximation of Gaussian processes is defined as
\begin{eqnarray}\label{FiniteGP}
Y^N(\boldsymbol{x})&=&\sum_{i_1,\ldots,i_d=0}^N Y(t_{N,i_1},\ldots,t_{N,i_d})\prod_{m\in\{1,\ldots,d\}}\phi_{i_{m}}(x_{m})\\ \nonumber
&=&\zeta_{i_1,\ldots,i_d}\prod_{m\in\{1,\ldots,d\}}\phi_{i_{m}}(x_{m}), \quad \boldsymbol{x}\in\mathcal{D}\subset\mathbb{R}^d,
\end{eqnarray}
where $(\zeta_{i_1,\ldots,i_d})=(Y(t_{N,i_1},\ldots,t_{N,i_d}))^\top$ is a zero-mean Gaussian vector with covariance matrix $\Gamma^N$ and $\phi_{i_{m}}$ is the hat function associated to the knots $t_{N,i_m}$: $\phi_{i_{m}}(x)=\phi((x-t_{N,i_m})/\Delta_N)$, where $\Delta_N=1/N$ and $\phi(x)=\left(1-| x|\right)\mathbb{1}_{(| x|\leq 1)}, \ x\in\mathbb{R}$. The value of any basis function at any knot is equal to Kronecker's Delta function ($\phi_{i_{m}}(t_{N,i_{m'}})=\delta_{i_m,i_{m'}}, \ i_m,i_{m'}=0,\ldots,N$), where $\delta_{i_m,i_{m'}}$ is equal to one if $i_m=i_{m'}$ and zero otherwise. The covariance function $K_N(\boldsymbol{x},\boldsymbol{x}')$ of the Gaussian process approximation $Y^N$ is equal to
\begin{equation*}
K_N(\boldsymbol{x},\boldsymbol{x}')=\Phi(\boldsymbol{x})^\top\Gamma^N\Phi(\boldsymbol{x}'),
\end{equation*}
where $\Phi(\boldsymbol{x})=(\prod_{m\in\{1,\ldots,d\}}\phi_{i_{m}}(x_{m}))_{i_m}$. This type of covariance functions are very similar to ones used in \citet{cressie2008fixed}, where $\Gamma^N$ is a square positive definite matrix estimated from the data, which it is not the case in the present paper. By this approach \eqref{FiniteGP}, simulate the GP approximation is equivalent to simulate the Gaussian vector $(\zeta_{i_1,\ldots,i_d})_{i_1,\ldots,i_d}$ restricted to 
\begin{eqnarray*}
&&Y^N(\boldsymbol{x}^{(i)})=y_i+\epsilon_i=\tilde{y}_i, \quad i=1,\ldots,n,\\ 
&&(\zeta_{i_1,\ldots,i_d})_{i_1,\ldots,i_d}\in C_{\text{coef}},
\end{eqnarray*}
where $\epsilon_i\overset{i.i.d.}\sim \mathcal{N}(0,\sigma_{\text{noise}}^2)$ and $C_{\text{coef}}$ is the space of coefficients which verify some linear constraints. Next, we show how $C_{\text{coef}}$ can be computed in each inequality constraints case. In this paper, boundedness, monotonicity and convexity constraints are considered but the methodology can be easily adapted to any convex sets. \

Notice that, model \eqref{FiniteGP} does not correspond to a truncated Karhunen-Lo\`eve expansion $Y(x)=\sum_{j=0}^{+\infty}Z_je_j(x)$; see, for example, \citet{Rasmussen2005GPM1162254,trecate1999finite} since the coefficients $\zeta_j$ are not independent (unlike the coefficients $Z_j$) and the basis functions $\phi_j$ are not the eigenfunctions $e_j$ of the Mercer kernel $K(x,x')$.

\subsection{Boundedness constraints}\label{BoundednessSection}
The one dimension is a particular case of the two dimensional one. The input $\boldsymbol{x}=(x_1,x_2)\in \mathbb{R}^2$ and without loss of generality in the unit square $\mathcal{D}=[0,1]^2$. The real function is supposed continuous and belong to the convex set
\begin{equation*}
C=\left\{f\in \mathcal{C}^0\left(\mathcal{D}\right)~: \ -\infty\leq a\leq f(\boldsymbol{x})\leq b\leq +\infty, \ \boldsymbol{x}\in \mathcal{D}\right\}.
\end{equation*}
The finite-dimensional approximation of GPs $\left(Y^N(\boldsymbol{x})\right)_{\boldsymbol{x}\in \mathcal{D}}$ is defined as
\begin{equation*}\label{boundedness2Dapproach}
Y^N(x_1,x_2)=\sum_{i,j=0}^NY(t_{N,i},t_{N,j})\phi_i(x_1)\phi_j(x_2)=\sum_{i,j=0}^N\zeta_{i,j}\phi_i(x_1)\phi_j(x_2),
\end{equation*}
where $\zeta_{i,j}=Y(t_{N,i},t_{N,j})$ and $(\phi_j)_j$ are the hat functions. Then, $Y^N$ is bounded between $a$ and $b$ with respect to the two inputs \textit{if and only if} the $(N+1)^2$ random coefficients $\zeta_{i,j}\in[a,b]$. 
This is because $Y^N$ is a piecewise-linear function. In that case, the space of inequality constraints on the coefficients is equal to
\begin{equation*}
C_{\text{coef}}=\left\{(\zeta_{i,j})_{i,j}\in\mathbb{R}^{(N+1)^2} \ : \ \zeta_{i,j}\in [a,b], \ i,j=0,\ldots,N\right\}.
\end{equation*}

\begin{rem}
The multidimensional case is an easy extension of two dimensional one. The input $\boldsymbol{x}\in\mathcal{D}\subset [0,1]^d$. The model defined in \eqref{FiniteGP} is bounded between $a$ and $b$ (i.e., $Y^N(\boldsymbol{x})\in[a,b]$) \textit{if and only if} the random coefficients $(\zeta_{i_1,\ldots,i_d})=(Y(t_{N,i_1},\ldots,t_{N,i_d}))\in[a,b]$. 
\end{rem}

\begin{prop}\label{boundaryproposition}
If the realizations of the original GP $Y$ are continuous, then the finite-dimensional approximation of GPs $Y^N$ is almost surely converge uniformly to $Y$ when $N$ tends to infinity.
\end{prop}

\begin{proof}
To prove the almost sure uniform convergence of the approximating random process $Y^N$ to the limiting process
$Y$, write more explicitly, for any $\omega\in\Omega$
\begin{equation*}
Y^N(x;\omega)=\sum_{j=0}^NY(t_{N,j};\omega)\phi_j(x), \quad x\in\mathcal{D}=[0,1].
\end{equation*}
Using the fact that $\phi_j(x)\geq 0$ and $\sum_{j=0}^N\phi_j(x)=1$, for all $x\in \mathcal{D}$, we get
\begin{eqnarray*}
\left|Y^N(x;\omega)-Y(x;\omega)\right|&=&\left|\sum_{j=0}^N(Y(t_{N,j};\omega)-Y(x;\omega))\phi_j(x)\right| \nonumber \\
&\leq & \sum_{j=0}^N \sup_{|x-x'|\leq \Delta_N}\left|Y(x';\omega)-Y(x;\omega)\right|\phi_j(x) \nonumber \\
&=&\sup_{|x-x'|\leq \Delta_N}\left|Y(x';\omega)-Y(x;\omega)\right|. \label{boundinequality}
\end{eqnarray*}
Thus, one can deduce that
\begin{equation*}
\sup_{x\in \mathcal{D}}\left|Y^N(x;\omega)-Y(x;\omega)\right|\underset{N\to +\infty}\longrightarrow 0
\end{equation*}
with probability 1, since the sample paths of the process $Y$ are uniformly continuous on the compact interval $\mathcal{D}$. 
\end{proof}

\subsection{Isotonicity constraints}\label{MTD}
The isotonicity constraints in two dimensions are considered. The input $\boldsymbol{x}=(x_1,x_2)\in\mathbb{R}^2$ and without loss of generality in the unit square $\mathcal{D}=[0,1]^2$. The real function $f$ is supposed to be monotone (non-decreasing) with respect to the two inputs
\begin{equation*}
x_1\leq x'_1 \mbox{\quad and \quad} x_2\leq x'_2 \quad \Rightarrow \quad f(x_1,x_2)\leq f(x'_1,x'_2).
\end{equation*}
The finite-dimensional approximation of GPs $\left(Y^N(\boldsymbol{x})\right)_{\boldsymbol{x}\in \mathcal{D}^2}$ is defined as
\begin{equation}\label{monotonicity2Dapproach}
Y^N(x_1,x_2)=\sum_{i,j=0}^NY(t_{N,i},t_{N,j})\phi_i(x_1)\phi_j(x_2)=\sum_{i,j=0}^N\zeta_{i,j}\phi_i(x_1)\phi_j(x_2),
\end{equation}
where $\zeta_{i,j}=Y(t_{N,i},t_{N,j})$ and $(\phi_j)_j$ are the hat functions. Then, $Y^N$ is non-decreasing with respect to the two inputs \textit{if and only if} the $(N+1)^2$ random coefficients $\zeta_{i,j}$ verify the following linear constraints: 
\begin{enumerate}
\item $\zeta_{i-1,j}\leq \zeta_{i,j} 
\mbox{ and } \zeta_{i,j-1}\leq \zeta_{i,j}, \ i,j=1,\ldots,N$;
\item $\zeta_{i-1,0}\leq \zeta_{i,0}, \ i=1,\ldots,N$;
\item $\zeta_{0,j-1}\leq \zeta_{0,j}, \ j=1,\ldots,N$.
\end{enumerate}

\begin{rem}[Isotonicity with respect to one variable]\label{rem2D}
If the function is non-decreasing with respect to the first variable only, then model \eqref{monotonicity2Dapproach} is non-decreasing with respect to $x_1$ \textit{if and only if} the random coefficients verify: $\zeta_{i-1,j}\leq \zeta_{i,j}, \ i=1,\ldots,N \ \mbox{and} \ j=0,\ldots,N$.
\end{rem}

\begin{rem}[Monotonicity constraints]
For monotonicity constraints, the finite-dimensional approximation of GPs can be written as 
\begin{equation}\label{monotone1D}
Y^N(x)=Y(0)+\sum_{j=0}^NY'(t_{N,j})I_j(x)=\gamma+\sum_{j=0}^N\zeta_jI_j(x), \quad x\in[0,1],
\end{equation}
where $\gamma=Y(0)$, $\zeta_j=Y'(t_{N,j})$ and $I_j(x)=\int_0^x\phi_j(t)dt$. In that case, $Y^N$ is monotone \textit{if and only if} the random coefficients $\zeta_j$ are all nonnegative. In fact, since $(I_j)$ are non-decreasing functions and $(\zeta_j)$ are nonnegative, then $Y^N$ is non-decreasing. Conversely, if $Y^N$ is non-decreasing, then $\zeta_j=(Y^N)'(t_{N,j})\geq 0$. Thus, the space of inequality constraints on the coefficients is 
{\em \begin{equation*}
C_{\text{coef}}=\left\{\left(\gamma,\zeta\right)\in \mathbb{R}^{N+2}~: \ \zeta_j\geq 0, \ j=0,\ldots,N\right\},
\end{equation*}}
where $\zeta=(\zeta_0,\ldots,\zeta_N)^\top$. For all $x,x'\in \mathcal{D}=[0,1]$, the covariance function of $Y^N$ is equal to
\begin{eqnarray*}
K_N(x,x')&=&\cov\left(Y^N(x),Y^N(x')\right)=K(0,0)+\sum_{i=0}^N\frac{\partial K}{\partial x}(t_{N,i},0)I_i(x)\\
&+&\sum_{j=0}^N\frac{\partial K}{\partial x'}(0,t_{N,j})I_j(x)+\sum_{i,j=0}^N\frac{\partial^2 K}{\partial x\partial x'}(t_{N,i},t_{N,j})I_i(x)I_j(x').
\end{eqnarray*}
\end{rem}

\subsection{Convexity constraints}
For convexity in two dimensions, the finite-dimensional approximation of GPs defined as
\begin{equation*}
Y^N(x_1,x_2)=\sum_{i,j=0}^NY(t_{N,i},t_{N,j})\phi_i(x_1)\phi_j(x_2)=\sum_{i,j=0}^N\zeta_{i,j}\phi_i(x_1)\phi_j(x_2),
\end{equation*}
is convex with respect to the two inputs \textit{if and only if} the random coefficients verify
\begin{enumerate}
\item $\frac{\zeta_{i,j}-\zeta_{i-1,j}}{t_{N,i}-t_{N,i-1}}\leq \frac{\zeta_{i+1,j}-\zeta_{i,j}}{t_{N,i+1}-t_{N,i}} 
\mbox{ and } \frac{\zeta_{i,j}-\zeta_{i,j-1}}{t_{N,j}-t_{N,j-1}}\leq \frac{\zeta_{i,j+1}-\zeta_{i,j}}{t_{N,j+1}-t_{N,j}}, \ i,j=1,\ldots,N-1$;
\item $\frac{\zeta_{i,0}-\zeta_{i-1,0}}{t_{N,i}-t_{N,i-1}}\leq \frac{\zeta_{i+1,0}-\zeta_{i,0}}{t_{N,i+1}-t_{N,i}}, \ i=1,\ldots,N-1$;
\item $\frac{\zeta_{0,j}-\zeta_{0,j-1}}{t_{N,j}-t_{N,j-1}}\leq \frac{\zeta_{0,j+1}-\zeta_{0,j}}{t_{N,j+1}-t_{N,j}}, \ j=1,\ldots,N-1$.
\end{enumerate}

\begin{rem}[Convexity in one dimension]\label{CC}
If the realizations of the original GP $Y$ are assumed to be at least twice differentiable. Then, the finite-dimensional approximation of GPs can be defined as 
\begin{equation*}\label{convexityapproach}
Y^N(x)=Y(0)+Y'(0)x+\sum_{j=0}^NY''(t_{N,j})\varphi_j(x)=\gamma+\kappa x+\sum_{j=0}^N\zeta_j\varphi_j(x),
\end{equation*}
where $\gamma=Y(0)$, $\kappa=Y'(0)$ and $\zeta_j=Y''(t_{N,j})$. The basis functions $(\varphi_j)_j$ are the two times primitive functions of $\phi_j$ 
\begin{equation*}
\varphi_j(x)=\int_0^x\left(\int_0^t\phi_j(u)du\right)dt, \quad x\in \mathcal{D}.
\end{equation*}
In that case, $Y^N$ is convex \textit{if and only if} the random coefficient $Y''(t_{N,j})$ are all nonnegative. Thus, the space of inequality constraints on the coefficients {\em $C_{\text{coef}}$} is equal to
{\em \begin{equation*}
C_{\text{coef}}=\left\{\left(\gamma,\kappa,\zeta\right)\in \mathbb{R}^{N+3}~: \ \zeta_j\geq 0, \ j=0,\ldots,N\right\},
\end{equation*}}
where $\zeta=(\zeta_0,\ldots,\zeta_N)^\top$. The covariance function of the GP approximation is equal to
\begin{equation*}
K_N(x,x')=\left(1,x,\varphi(x)^\top\right)\tilde{\Gamma}^N\left(1,x',\varphi(x')^\top\right)^\top, 
\end{equation*}
where $\varphi(x)=\left(\varphi_0(x),\ldots,\varphi_N(x)\right)^\top$ and 
\begin{equation*}
\tilde{\Gamma}^N=\left[\begin{matrix}
K(0,0) & \frac{\partial K}{\partial x'}(0,0) & \frac{\partial^2 K}{\partial (x')^2}(0,t_{N,j})\\
\frac{\partial K}{\partial x}(0,0) & \frac{\partial^2K}{\partial x\partial x'}(0,0) & 
\frac{\partial^3K}{\partial x\partial (x')^2}(0,t_{N,j})\\
\frac{\partial^2 K}{\partial x^2}(t_{N,i},0) & \frac{\partial^3K}{\partial x^2\partial x'}(t_{N,i},0) &\Gamma^N_{i,j}
\end{matrix}
\right]_{0\leq i,j\leq N},
\end{equation*}
\noindent and
\begin{equation*}
\Gamma^N_{i,j}=\cov(Y''(t_{N,i}),Y''(t_{N,j}))=\frac{\partial^4K}{\partial x^2\partial (x')^2}(t_{N,i},t_{N,j}), \quad i,j=0,\ldots,N.
\end{equation*}
\end{rem}

\subsection{Simulated paths}
This subsection is devoted to the sampling scheme of the proposed model conditionally to inequality constraints and noisy observations. To simplify notations, the finite-dimensional approximation of GPs in one dimension is considered
\begin{equation*}
Y^N(x)=\sum_{j=0}^NY(t_{N,j})\phi_j(x)=\sum_{j=0}^N\zeta_j\phi_j(x), \quad x\in\mathcal{D}.
\end{equation*}
In this paper, the GP is observed with error. The space of noisy observations is defined as
\begin{eqnarray*}
I_{\text{coef}}&=&\left\{\zeta\in\mathbb{R}^{N+1}~:  \sum_{j=0}^N\zeta_j\phi_j(x^{(i)})=\tilde{y}_i, \ i=1,\ldots,n\right\}\\
&=&\left\{\zeta\in\mathbb{R}^{N+1}~:  A\zeta=\boldsymbol{\tilde{y}}\right\},
\end{eqnarray*}
where $\tilde{y}_i=y_i+\epsilon_i, \ i=1,\ldots,n$, $\epsilon_i\overset{i.i.d.}\sim \mathcal{N}(0,\sigma_{\text{noise}}^2)$ and $A_{i,j}=\phi_j(x^{(i)})$.
The set of inequality constraints on the coefficients $C_{\text{coef}}$ is a convex set (for instance, the nonnegative quadrant $\zeta_j\geq 0, \ j=0,\ldots,N$ for non-decreasing constraints in one dimension). The sampling scheme can be summarized in two steps: first, the conditional Gaussian vector $\zeta$ with only noisy observations is simulated
\begin{equation*}\label{intXi}
\footnotesize
\zeta \suchthat A\zeta=\boldsymbol{\tilde{y}} \sim \mathcal{N}\left((A\Gamma^N)^\top(A\Gamma^NA^\top+\sigma_{\text{noise}}^2\textbf{I})^{-1}\boldsymbol{\tilde{y}},\Gamma^N-(A\Gamma^N)^\top (A\Gamma^NA^\top+\sigma_{\text{noise}}^2\textbf{I})^{-1}A\Gamma^N\right).
\end{equation*}
Second, by an improved rejection sampling \citep{Maatouk2016}, only the random coefficients in the convex set $C_{\text{coef}}$ are selected. Now, the three estimates used in the illustrative examples (Section~\ref{IllExamples}) are defined.

\begin{defi}\label{unconstrMean}
The so-called unconstrained mean is defined as
{\em \begin{eqnarray*}
m^N(x)=E\left(Y^N(x) \suchthat Y^N(x^{(i)})=\tilde{y}_i, \ i=1,\ldots,n\right)=\phi(x)^\top\zeta_{\text{I}},
\end{eqnarray*} }
where {\em $\zeta_{\text{I}}=E\left(\zeta \suchthat \zeta\in I_{\text{coef}}\right)=\Gamma^NA^\top\left(A\Gamma^NA^\top+\sigma_{\text{noise}}^2\textbf{I}\right)^{-1}\boldsymbol{\tilde{y}}$}. 
\end{defi}
Similarly to the kriging mean of the original GP $Y$ (Eq.~\eqref{condGP}, $Z=Y$ when $\eta$ is the null function), the kriging mean $m^N$ of the finite-dimensional approximation of GPs $Y^N$ can be written as
\begin{equation*}
m^N(x)=\boldsymbol{k}_N(x)^\top\left(\mathbb{K}_N+\sigma_{\text{noise}}^2\textbf{I}\right)^{-1}\boldsymbol{\tilde{y}},
\end{equation*}
where $\boldsymbol{k}_N(x)=K_N(x,\boldsymbol{X})=\left(A\Gamma^N\phi(x)\right)$ is the vector of covariance between $Y^N(x)$ and $Y^N\left(\boldsymbol{X}\right)$ and $(\mathbb{K}_N)_{i,j}=K_N(x^{(i)},x^{(j)})=\left(A\Gamma^NA^\top\right)_{i,j}$, $i,j=1,\ldots,n$ is the covariance matrix of $Y^N\left(\boldsymbol{X}\right)$.  

\begin{rem}\label{unconstrained_remark}
The unconstrained mean $m^N(x)$ respects inequality constraints in the entire domain \textit{if and only} if the conditional Gaussian vector to only noisy observations {\em $\zeta_{\text{I}}$} lies inside the convex set {\em $C_{\text{coef}}$}.
\end{rem}
 
\begin{defi}\label{IKMdef}
The mean of the posterior distribution of $Y^N$ conditionally to inequality constraints and noisy observations is defined as 
{\em \begin{equation*}
m^N_{\text{pos}}(x)=E\left(Y^N(x)\suchthat Y^N(x^{(i)})=\tilde{y}_i, \ \zeta\in C_{\text{coef}}\right)=\phi(x)^\top\zeta_{\text{pos}},
\end{equation*}} 
where {\em $\zeta_{\text{pos}}=E\left(\zeta \suchthat \zeta\in I_{\text{coef}}\cap C_{\text{coef}}\right)$} is the mean of the truncated Gaussian vector which is computed from simulations.
\end{defi}

Finally, let $\mu$ be the maximum of the probability density function (pdf) of $\zeta$ restricted to $I_{\text{coef}}\cap C_{\text{coef}}$. It is the solution of the following convex optimization problem
\begin{equation}\label{mu}
\mu=\arg\min_{x\in I_{\text{coef}}\cap C_{\text{coef}}}\left(\frac{1}{2}x^\top\left(\Gamma^N\right)^{-1}x\right),
\end{equation}
where $\Gamma^N$ is the covariance matrix of the Gaussian vector $\zeta$. The quadratic optimization problem \eqref{mu} is equivalent to 
\begin{equation}\label{muBIS}
\mu=\arg\min_{x\in C_{\text{coef}}}\left(\frac{1}{2}x^\top\left(\Gamma_{\text{cond}}^N\right)^{-1}x+\zeta_{\text{I}}^\top x\right),
\end{equation}
where $\Gamma_{\text{cond}}^N$ is the covariance matrix of the conditional Gaussian vector $\zeta \suchthat A\zeta=\boldsymbol{\tilde{y}}$. In fact, $\mu$ representes the maximum of the pdf of the Gaussian vector $\zeta$ restricted to $I_{\text{coef}}\cap C_{\text{coef}}$ and its numerical calculation is a standard problem in the minimization of positive quadratic forms subject to convex constraints \citep{Boyd2004CO993483,Goldfarb83}. Let us mention that in all simulated examples illustrated in this paper, the R-package `solve.QP' described in \citet{Goldfarb83} is used to solve the quadratic convex optimization problems \eqref{mu}-\eqref{muBIS}.

\begin{defi}\label{Mdef}
The maximum of the posterior distribution of $Y^N$ conditionally to inequality constraints and noisy observations is defined as
{\em \begin{equation*}
M^N_{\text{pos}}(x)=\sum_{j=0}^N\mu_j\phi_j(x),\quad x\in\mathbb{R}^d,
\end{equation*} }
where $\mu=(\mu_0,\ldots,\mu_N)^\top$ is computed by \eqref{muBIS}. 
\end{defi}

\begin{rem}
The maximum \textit{a posteriori} estimate {\em $M^N_{\text{pos}}$} does not depend on the variance hyper-parameter $\sigma$ of the covariance function $K$ as well as on the simulations but depends on the length hyper-parameters of the covariance function $\boldsymbol{\theta}=(\theta_1,\ldots,\theta_d)$.  
\end{rem}

\begin{rem}\label{correspondRemark}
In the case where the GP is observed without error (i.e., with noise-free data), the maximum \textit{a posteriori} estimate {\em $M^N_{\text{pos}}$} converges uniformly to the constrained interpolation function solution of the following convex optimization problem
\begin{equation*}
\arg\min_{h\in H\cap I\cap C}\|h\|_H^2,
\end{equation*}
where $H$ is the reproducing kernel Hilbert space (RKHS) associated to the positive type kernel $K$ \citep{aronszajn1950,berlinet2011reproducing}, $I$ is the set of functions verify interpolation conditions and the convex set $C$ is the space of functions verify inequality constraints \citep{bay2016,Bay2017}. \

This generalizes to the case of interpolation conditions and inequality constraints the well known correspondence established by \citet{KW1970} between Bayesian estimation on stochastic process and smoothing by splines. 
\end{rem}

\begin{algorithm}[t]
  \caption{Sampling scheme}
  \label{alg:rejection_sampling}
 \DontPrintSemicolon
 \textbf{Initialization}: \\
  $\zeta\notin C_{\text{coef}}$; $\zeta \gets \zeta_{\text{current}}$ \\
  $1 \gets \text{unif}$; $0 \gets t$ \\
  \While{{\em $\text{unif}>t$}}{
  $\zeta \gets \zeta_{\text{current}}$ \\
  \While{{\em $\zeta_{\text{current}}\notin C_{\text{coef}}$}}{
 $ \mathcal{N}(\mu,\Gamma^N_{\text{cond}}) \gets \zeta_{\text{current}}$
 }
$ \exp\left(\mu^\top(\Gamma^N_{\text{cond}})^{-1}(\mu-\zeta_{\text{I}}-\zeta_{\text{current}})+\zeta_{\text{current}}^\top(\Gamma^N_{\text{cond}})^{-1}\zeta_{\text{I}}\right) \gets t$ \\
$\mathcal{U}(0,1)\gets \text{unif}$
}
\end{algorithm}

In Algorithm~\ref{alg:rejection_sampling}, the sampling scheme of the proposed model is described. It is based on the rejection sampling from the Mode (RSM) algorithm to simulate the Gaussian vector $\zeta$ restricted to the convex set $I_{\text{coef}}\cap C_{\text{coef}}$ (see, \citet{Maatouk2016} for more details).

\section{Illustrative examples}\label{IllExamples}
The goal of this section is twofold: first, to illustrate the condition simulation of the GP approximation developed in the present paper with certain constraints such as boundedness, positivity and monotonicity in one and two dimensions and noisy observations. Second, to describe the two different cases in the simulation.
\begin{itemize}
\item The unconstrained mean respects the constraints and then coincides with the maximum of the posterior distribution. 
\item The unconstrained mean does not respect the constraints, then the unconstrained mean and the maximum of the posterior distribution are different.
\end{itemize}
The Mat\'ern 3/2 and squared exponential (or Gaussian) covariance functions are used (Table~\ref{kernels}).

\subsection{Boundedness constraints}
The real function is supposed to respect boundedness constraints
\begin{equation}\label{BoundedSpace}
C=\left\{f\in\mathcal{C}^0([0,1]) \ : \ -\infty\leq a\leq f(x)\leq b\leq +\infty, \ x\in[0,1]\right\}.
\end{equation}
The constrained data of size $n=10$ (black points in Fig.~\ref{postive_bounded}) are not taken from constrained functions. The noise variance is fixed to $\sigma_{\text{noise}}^2=1.1^2$. Additionally, the Mat\'ern 3/2 covariance function is used with the hyper-parameters fixed to $(\theta,\sigma)=(0.3,10)$. 

\begin{figure}[h]
\begin{minipage}{.5\linewidth}
\centering
\subfloat[]{\label{positive}\includegraphics[scale=.3]{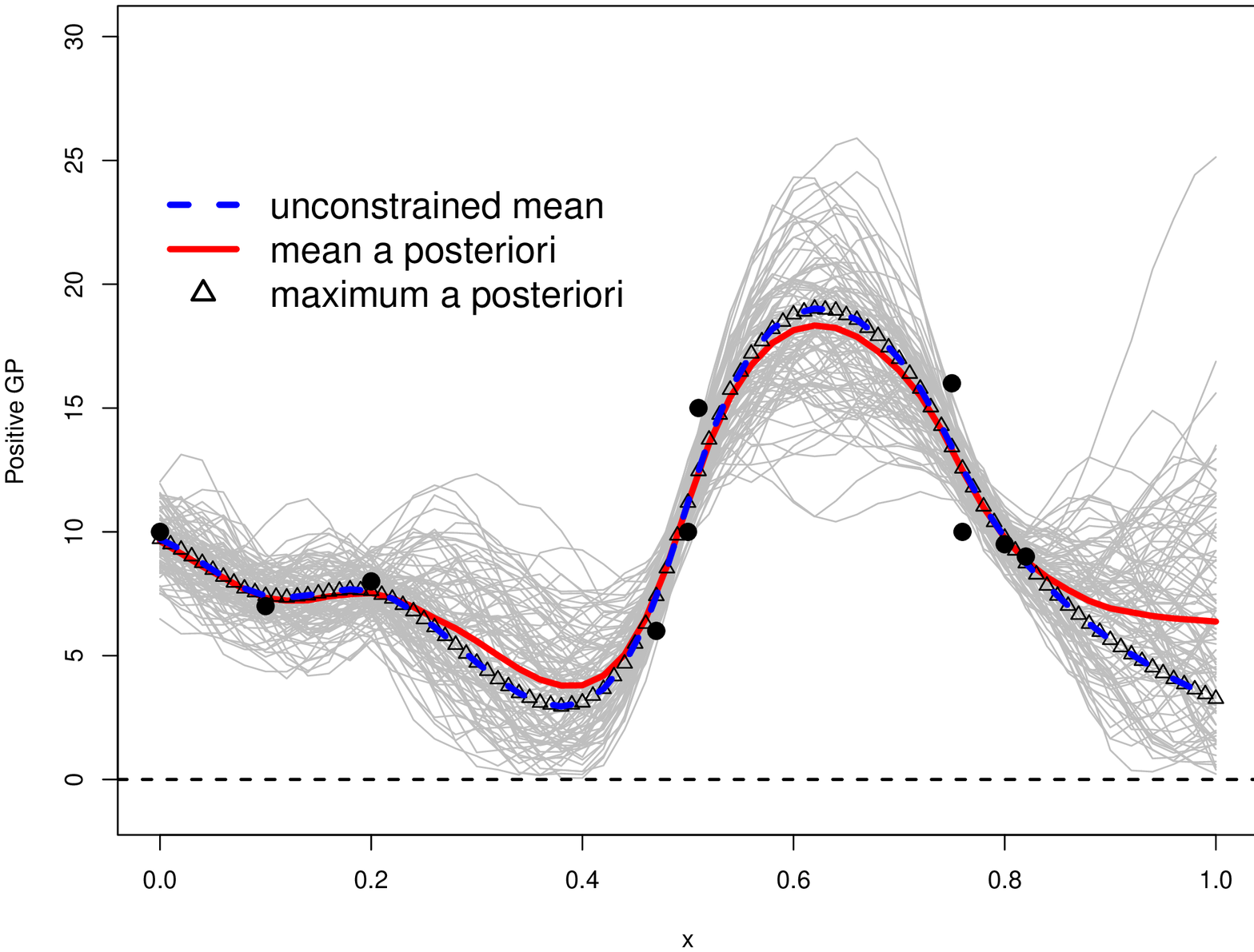}}
\end{minipage}%
\begin{minipage}{.5\linewidth}
\centering
\subfloat[]{\label{bornee}\includegraphics[scale=.3]{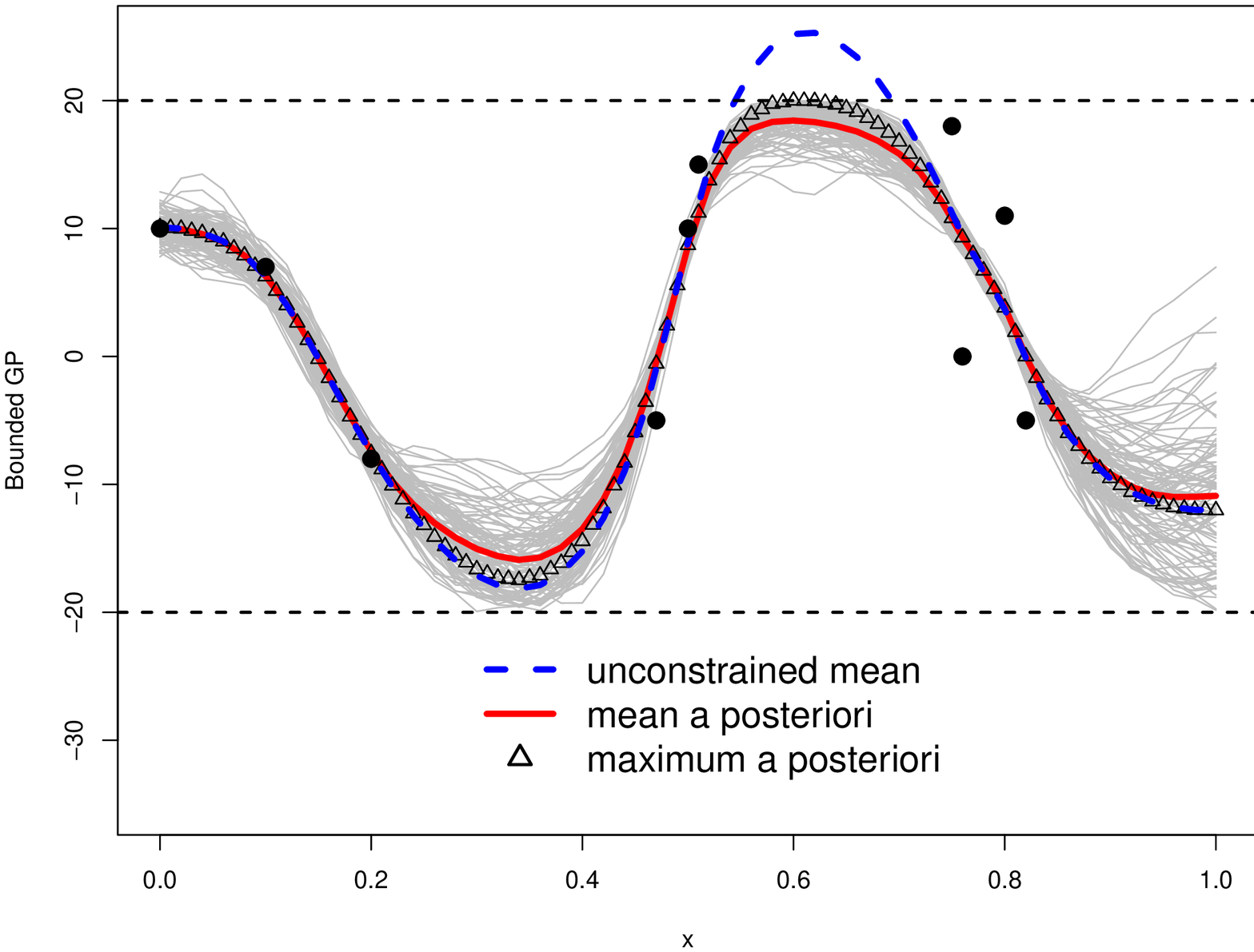}}
\end{minipage}
\caption{The GP approximation with positivity constraints (a) and boundedness constraints (b). The unconstrained mean coincides with the maximum \textit{a posteriori} in (a) but not in (b).}
\label{postive_bounded}
\end{figure}

In Fig.~\ref{positive}, we generate one hundred sample paths taken from model~\eqref{FiniteGP} with $d=1$ and $N=50$ conditionally to positivity constraints (i.e., $a=0$ and $b=+\infty$ in \eqref{BoundedSpace}). The simulated trajectories (gray lines) respect positivity constraints in the entire domain as well as the mean of the posterior distribution. The unconstrained mean and the maximum of the posterior distribution coincide and respect positivity constraints in the entire domain as well: it corresponds to the situation where the conditional Gaussian vector $\zeta_{\text{I}}$ lies inside the acceptance region $C_{\text{coef}}$ (Remark~\ref{unconstrained_remark}). In Fig.~\ref{bornee}, the boundedness constraint is considered (i.e., $a=-20$ and $b=20$ in \eqref{BoundedSpace}). The simulated trajectories (gray lines) respect boundedness constraints in the entire domain as well as the mean  and the maximum of the posterior distribution, contrarily to the unconstrained mean. This is the case where $\zeta_{\text{I}}$ lies outside the acceptance region $C_{\text{coef}}$ (Remark~\ref{unconstrained_remark}). This numerical result can be seen as a generalization of the Kimeldorf-Wahba correspondence \citet{KW1970} in the case of inequality constraints and errors measurements between Bayesian estimation on stochastic process and smoothing by splines.

\subsection{Monotonicity constraints}
The monotone (non-decreasing) function $f(x)=0.32(x+sin(x)), \ x\in [0,10]$ used in the literature to compare different models is considered. It is evaluated at data of size $n=50$ chosen randomly on $[0,10]$ (black points in Fig.~\ref{monotone}) with standard deviation $\sigma_{\text{noise}}=1$. 

\begin{figure}[h]
\centering
\includegraphics[scale=.35]{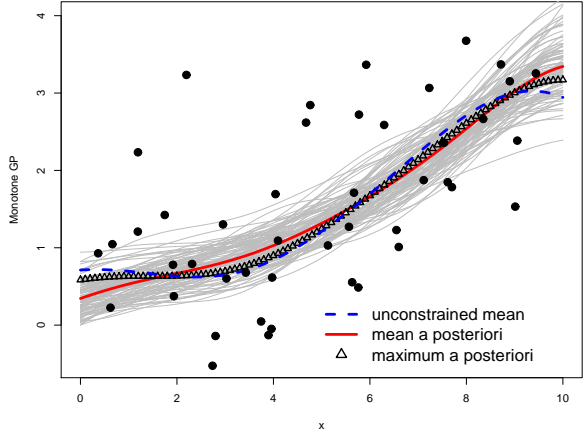}
\caption{The GP approximation~\eqref{monotone1D} with monotonicity constraints for sinusoidal function $f(x)=0.32(x+sin(x))$. The unconstrained mean does not coincide with the maximum \textit{a posteriori}.}
\label{monotone}
\end{figure}

In Fig.~\ref{monotone}, we generate one hundred sample paths taken from model~\eqref{monotone1D} with $N=50$ conditionally to monotonicity (non-decreasing) constraints. The squared exponential covariance function is used with hyper-parameters $(\theta,\sigma)=(2.5,1)$. Notice that, the simulated trajectories (gray lines) are non-decreasing in the entire domain as well as the mean and the maximum of the posterior distribution, contrarily to the unconstrained mean. It corresponds to the case where the conditional Gaussian vector $\zeta_{\text{I}}$ lies outside the acceptance region $C_{\text{coef}}$ (Remark~\ref{unconstrained_remark}).

\subsection{Isotonicity in two dimensions}
In two dimensions, the monotone (non-decreasing) function with respect to the two inputs used in \citet{saarela2011method,shively2009bayesian} 
\begin{equation*}
f(x_1,x_2)=\mathbb{1}_{\{(x-1-1)^2+(x_2-1)^2<1\}}\{1-(x_1-1)^2-(x_2-1)^2\}^{1/2}, \quad (x_1,x_2)\in [0,1]^2
\end{equation*}
is considered. It is evaluated at data of size $n=100$ chosen randomly on $[0,1]^2$ with standard deviation $\sigma_{\text{noise}}=0.1$. In Fig.~\ref{monotone2Dcontour}, the two-dimensional squared exponential covariance function is used
\begin{equation}\label{SE2D}
K(\boldsymbol{x},\boldsymbol{x}')=\sigma^2\exp\left(-\frac{(x_1-x'_1)^2}{2\theta_1^2}\right)\times\exp\left(-\frac{(x_2-x'_2)^2}{2\theta_2^2}\right),
\end{equation}
where the variance hyper-parameter $\sigma=1$ and the length hyper-parameters $(\theta_1,\theta_2)=(0.02,0.17)$ are estimated using cross-validation methods \citep{Maatouk201538}. Figure~\ref{monotone2Dcontour} shows the maximum of the posterior distribution using model~\eqref{monotonicity2Dapproach} with $N=10$ and the associated contour levels. It respects monotonicity (non-decreasing) constraints with respect to the two inputs.

\begin{figure}[h]
\begin{minipage}{.5\linewidth}
\centering
\subfloat[]{\includegraphics[scale=.3]{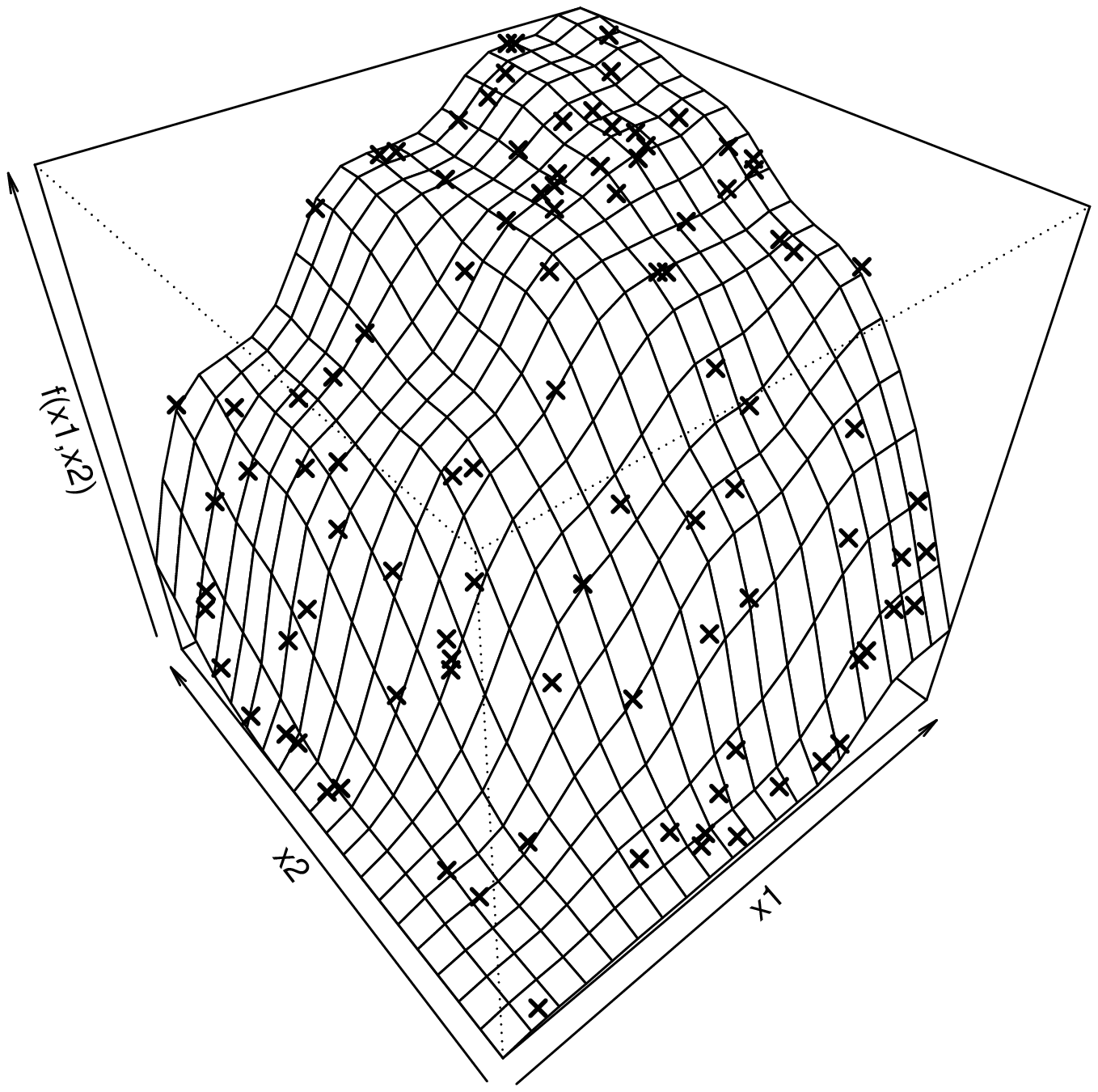}}
\end{minipage}%
\begin{minipage}{.5\linewidth}
\centering
\subfloat[]{\includegraphics[scale=.25]{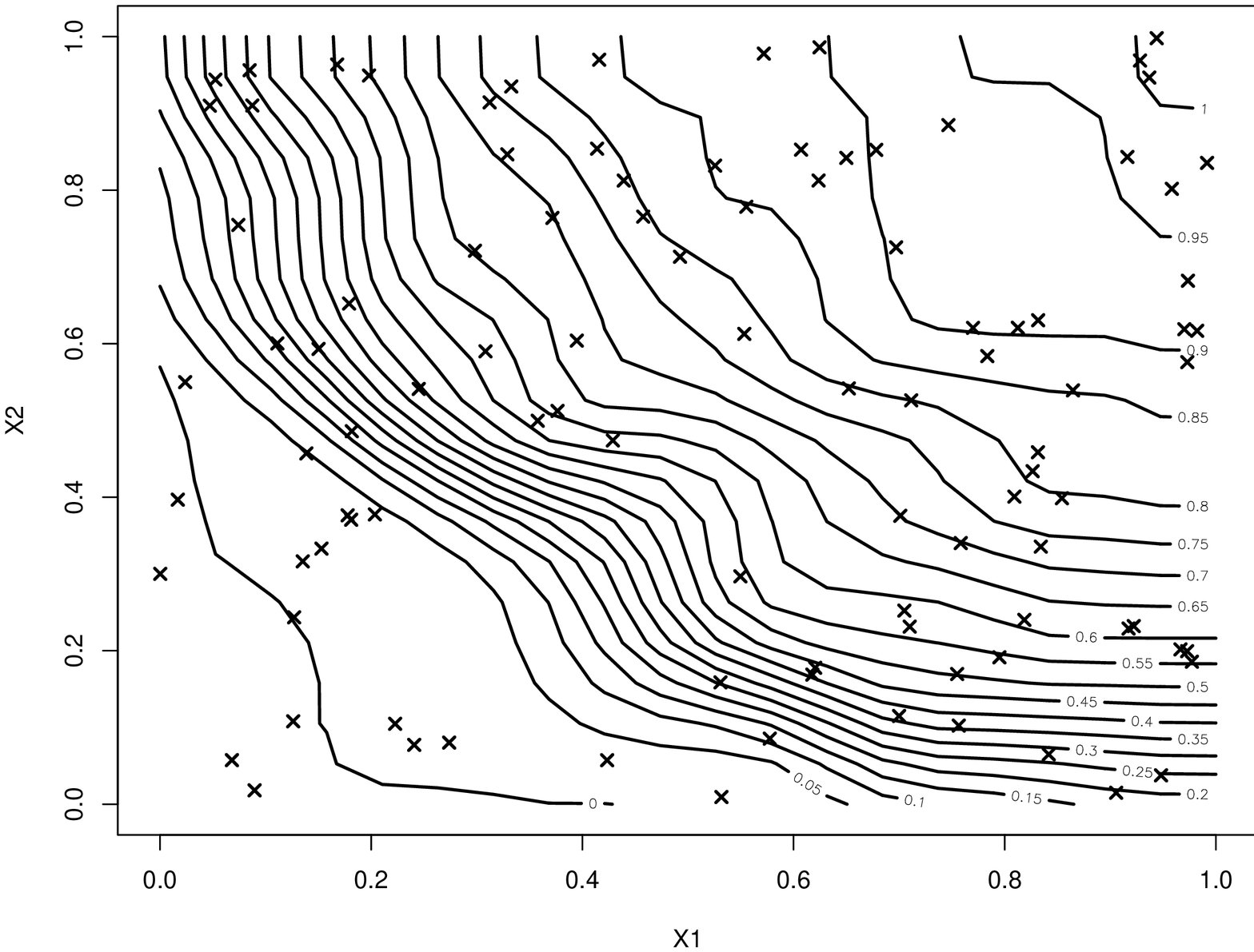}}
\end{minipage}
\caption{The maximum of the posterior distribution drawn from model~\eqref{monotonicity2Dapproach} respecting monotonicity (non-decreasing) constraints for the two inputs, and the associated contour levels.}
\label{monotone2Dcontour}
\end{figure}

\begin{rem}
For monotonicity with respect to only one variable, model~\eqref{monotonicity2Dapproach} (with noise-free data) has been used in \citet{Cousin2016} to estimate the discount factor surface as a function of time-to-maturities and quotation dates. It is a monotone (non-increasing) function with respect to time-to-maturities at each quotation date.
\end{rem}

\section{Simulation study}\label{SimStudy}
In this section, a comparison between the finite-dimensional approximation of GPs developed in the present paper and models deal with monotonicity and isotonicity constraints is shown. The real non-decreasing functions proposed by \citet{holmes2003generalized,neelon2004bayesian} and used in a comparative study by \citet{shively2009bayesian,lin2014bayesian} are considered
\begin{itemize}
\item flat function $f_1(x)=3, \ x\in (0,10]$;
\item sinusoidal function $f_2(x)=0.32\{x+sin(x)\}, \ x\in (0,10]$;
\item step function $f_3(x)=3$ if $x\in (0,8]$ and $f_3(x)=8$ if $x\in (8,10]$;
\item linear function $f_4(x)=0.3x, \ x\in (0,10]$;
\item exponential function $f_5(x)=0.15\exp(0.6x-3), \ x\in(0,10]$;
\item logistic function $f_6(x)=3/\{1+\exp(-2x+10)\}, \ x\in (0,10]$.
\end{itemize}

These functions are supposed to be evaluated at data of size $n=100$ with standard deviation $\sigma_{\text{noise}}=1$. The root-mean-square error (RMSE) of the estimates is computed at the one hundred $x$ values taken uniformly (equidistant) in the interval $(0,10]$:
\begin{equation*}
\text{RMSE}=\sqrt{\frac{1}{n}\sum_{i=1}^n\left(f(x_i)-\hat{f}(x_i)\right)^2}, 
\end{equation*}
where $\hat{f}(x)$ is the estimate of $f(x)$ and $x_i$ are the $n$ equally-spaced $x$-values. For the GP approximation developed in this paper, the maximum \textit{a posteriori} estimate (Definition~\ref{Mdef}) is used as an estimate of $f(x)$, where $N$ is fixed to fifty. Let us recall that this estimate depends only on the length hyper-parameter $\theta$. The squared exponential covariance function (Table~\ref{kernels}) is used in the simulation, with $\sigma$ fixed to 1 and $\theta$ estimated using the suited cross-validation method \citep{Maatouk201538,Cousin2016}. Table~\ref{ThetaEst} shows the values of the parameter estimation $\hat{\theta}$.

\begin{table}
\centering
\caption{Length hyper-parameter estimates using a suited cross-validation method.}
\begin{tabular}{*{7}{c}}
  &  Flat  & Step & Linear  & Exponential & Logistic & Sinusoidal    \\ \hline
$\hat{\theta}$     & 100.0 	& 0.8 & 8.6 &	1.0 & 2.0 & 2.5    \\
\end{tabular}
\label{ThetaEst}
\end{table}

In Table~\ref{RMSE}, the RMSE of the estimates is calculated for the finite-dimensional approximation of GPs, and it is compared with results of Gaussian process with and without projection given in \citet{lin2014bayesian} and results of the regression spline method given in \citet{shively2009bayesian}. To ensure stability of results, the simulations have been repeated 5000 times. Table~\ref{RMSE} shows that the finite-dimensional approximation of GPs outperforms regression splines (resp. Gaussian process with and without projection) except in the step and logistic cases (resp. in the step and exponential cases). 

\begin{table}
\centering
\caption{Root-mean-square error ($\times$ 100) for data of size $n=100$. The results are obtained by repeating the simulation 5000 times.}
\begin{tabular}{*{7}{c}}
 &  Flat  & Step & Linear  & Exponential & Logistic & Sinusoidal  \\  \hline  
Gaussian process         & 15.1 	& 27.1 & 16.7 &	19.7 & 25.5 & 21.9    \\
Gaussian process projection        & 11.3 & 25.3 & 16.3 &	19.1 & 22.4  & 21.1     \\
Regression spline         & 9.7  & 28.5 & 24.0 &	21.3   & 19.4  & 22.9 \\
Gaussian process approximation         & 8.2  & 41.1 & 15.8  & 20.8	&  21.0 & 20.6    \\
\end{tabular}
\label{RMSE}
\end{table}

\begin{rem}
Let us recall that the finite-dimensional approximation of GPs developed in the present paper is supposed centered (i.e., mean-zero). To be coherent, the results presented in Table~\ref{RMSE} should be computed when the output values are normalized 
\begin{equation*}
\bar{y}_i=y_i-\bar{y}, \quad i=1,\ldots,n,
\end{equation*}
where $\bar{y}_i$ is the normalized value of the $i$th output observation and $\bar{y}=1/n\sum_{i=1}^ny_i$ is the mean of the output observations. In that case, the finite-dimensional approximation of GPs outperforms regression spline and Gaussian process with and without projection except in the step case. 
\end{rem}

Now, the uncertainty quantification is investigated. The monotone (non-decreasing) function $f(x)=0.32(x+sin(x))$, $x\in (0,10]$ (sinusoidal function) is considered (dashed lines in Fig.~\ref{monot_sinusoidal}). It is evaluated at data of size $n=100$ distributed randomly on $(0,10]$ (grey crosses in Fig.~\ref{monot_sinusoidal}), with standard deviation $\sigma_{\text{noise}}=1$.

\begin{figure}
\centering
\includegraphics[scale=0.35]{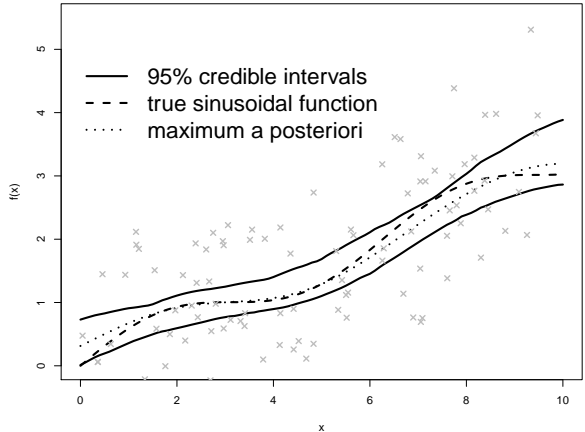}
\caption{The 95\% credible intervals of the Gaussian process approximation together with the sinusoidal function, the observations (grey crosses) and the maximum \textit{a posteriori} estimate.}
\label{monot_sinusoidal}
\end{figure}

In Table~\ref{Converage}, the percentage of the empirical coverage of 95\% pointwise credible intervals of GP approximation is computed by repeating the simulation 1000 times. The coverage for Gaussian process approximation is closer to the nominal 95\% than is that of the Gaussian process at most of input locations chosen by \cite{lin2014bayesian}. Additionally, the finite-dimensional approximation of GPs outperforms Gaussian process with projection at some input locations and slightly bad at the other locations.

\begin{table}
\centering
\caption{Empirical coverage (\%) for 95\% credible intervals at different $x$ values. The simulations are repeated 1000 times.}
\resizebox{0.85\textwidth}{!}{
\begin{minipage}{\textwidth}
\begin{tabular}{*{11}{c}}
 &  0.5  & 1 & 1.5  & 2 & 2.5 & 3 & 3.5 & 4 & 4.5 & 5  \\  \hline  
Gaussian process		& 97.3 & 94.6 & 91.8 & 88.0  & 90.5 & 95.2	 & 96.8 & 91.0 & 86.5 & 86.3 \\
Gaussian process projection        & 94.1 & 95.4 &  92.0 &  89.5 & 93.1 & 94.6 & 96.0 & 90.0 &  89.0 &  86.9  \\
Gaussian process approximation    & 97.0 & 93.0  & 89.6  &  90.1  &  94.1  & 97.1  &  95.5  & 89.5	& 85.4  & 86.7  \\
\end{tabular}
\label{Converage}
\end{minipage}}
\end{table}

To compare the proposed approach with the methodology based on the knowledge of the derivatives of the GP at some input locations, the logistic artificial function $f(x)=2/(1+\exp(-8x+4)), \ x\in [0,1]$ defined in \citet{journals/jmlr/RiihimakiV10} is considered. This function is supposed to be evaluated at data of size $n$ with standard deviation $\sigma_{\text{noise}}=0.5$. The squared exponential covariance function is used. 

\begin{figure}
\centering
\includegraphics[scale=0.35]{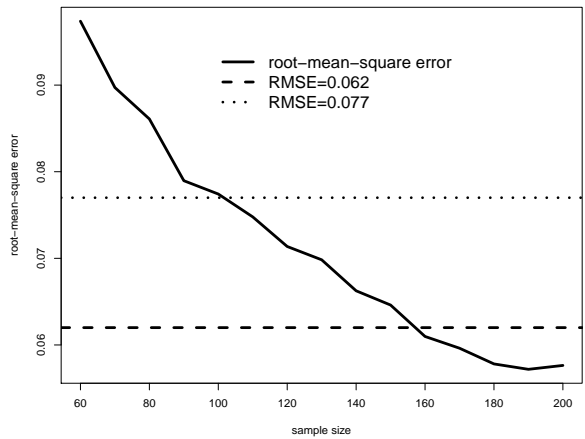}
\caption{The root-mean-square error at different sample sizes together with the optimal values obtained in \citet{journals/jmlr/RiihimakiV10}.}
\label{conv_RMSE_sample_size}
\end{figure}

In \citet{journals/jmlr/RiihimakiV10}, the RMSE is equal to $0.077$ (resp. $0.062$) for $n=100$ (resp. $n=200$). In Fig.~\ref{conv_RMSE_sample_size}, the root-mean-square error using the GP approximation is illustrated at different sample sizes together with the optimal values obtained by \citet{journals/jmlr/RiihimakiV10}. Notice that, we just need data of size $n=160$ to reach the optimal value $0.062$ obtained by \citet{journals/jmlr/RiihimakiV10}. The results are based on 1000 simulation replicates. \\

The isotonicity (non-decreasing) functions with respect to the two inputs used in \citet{lin2014bayesian,saarela2011method} are considered
\begin{eqnarray*}
&&  f_1(x_1,x_2)=\sqrt{x_1}, \ (x_1,x_2)\in [0,1]^2;\\
&& f_2(x_1,x_2)=0\text{.}5x_1+0\text{.}5x_2, \ (x_1,x_2)\in [0,1]^2;\\
&& f_3(x_1,x_2)=\min(x_1,x_2), \ (x_1,x_2)\in [0,1]^2;\\
&& f_4(x_1,x_2)=0\text{.}25x_1+0\text{.}25x_2+0\text{.}5\times\mathbb{1}_{\{x_1+x_2>1\}}, \ (x_1,x_2)\in [0,1]^2;\\
&& f_5(x_1,x_2)=0\text{.}25x_1+0\text{.}25x_2+0\text{.}5\times\mathbb{1}_{\{\min(x_1,x_2)>5\}}, \ (x_1,x_2)\in [0,1]^2;\\
&& f_6(x_1,x_2)=\mathbb{1}_{\{(x_1-1)^2+(x_2-1)^2<1\}}\sqrt{1-(x_1-1)^2-(x_2-1)^2}, \ (x_1,x_2)\in [0,1]^2.
\end{eqnarray*}
The two-dimensional squared exponential covariance function \eqref{SE2D} is used, with $\sigma$ fixed to 1 and $(\theta_1,\theta_2)$ estimated using the suited cross-validation method \citep{Maatouk201538,Cousin2016}. Table~\ref{ThetaEst2D} shows the values of the parameter estimation $(\hat{\theta}_1,\hat{\theta}_2)$.

\begin{table}
\centering
\caption{The summary of length hyper-parameter estimates in two dimensions using cross-validation methods.}
\begin{tabular}{*{7}{c}}
 &  $f_1$  & $f_2$ &  $f_3$  &  $f_4$ &  $f_5$ & $f_6$  \\   \hline 
$(\hat{\theta}_1,\hat{\theta}_2)$  & (0.17,0.38) & (0.46,1.32) & (0.18,0.22) &(0.38,0.01) & (0.08,0.09) & (0.02,0.17)        \\
\end{tabular}
\label{ThetaEst2D}
\end{table}

\begin{table}
\centering
\caption{Mean square error ($\times$ 100) for data of size $n=1024$ with standard deviation $\sigma_{\text{noise}}=0.1$. The results are based on 100 simulation replicates.}
\begin{tabular}{*{7}{c}}
 &  $f_1$  & $f_2$ & $f_3$  & $f_4$ & $f_5$ & $f_6$  \\    \hline
Gaussian process projection        & 0.04 & 0.02 & 0.05 &	0.20 & 0.19  &  0.10  \\
Gaussian process approximation         &  2.86e-3  & 4.40e-4  &  7.09e-3  &  0.55	&  0.34  & 0.04    \\
\end{tabular}
\label{MSE}
\end{table}

In Table~\ref{MSE}, the mean square error (MSE) of the estimates is calculated for the finite-dimensional approximation of GPs, and it is compared with results of Gaussian process projections given in \citet{lin2014bayesian}. Table~\ref{MSE} shows that the finite-dimensional approximation of GPs outperforms Gaussian process projections except in $f_4$ and $f_5$ cases. This is very similar to the one-dimensional case results, because of the similarity of $f_4$ and $f_5$ functions to the step case.

\section{Real application: nuclear safety}
In this section, the performance of the proposed model has been investigated using the real-word data provided by the `Institut of Radioprotection and Nuclear Safety' (IRSN), France. The nuclear reactor of the uranium sphere called `Lady Godiva device' situated at Los Alamos National Laboratory (LANL), New Mexico, U.S. has been studied. The nuclear reactor of the sphere is increasing with respect to the two considered input parameters: its radius (between 0 and 20 cm) and density (between 10 and 20 g/cm$^3$).
\begin{figure}[hptb]
  \centering
  \subfloat[][\label{godivaData}]{\includegraphics[width=.4\textwidth]{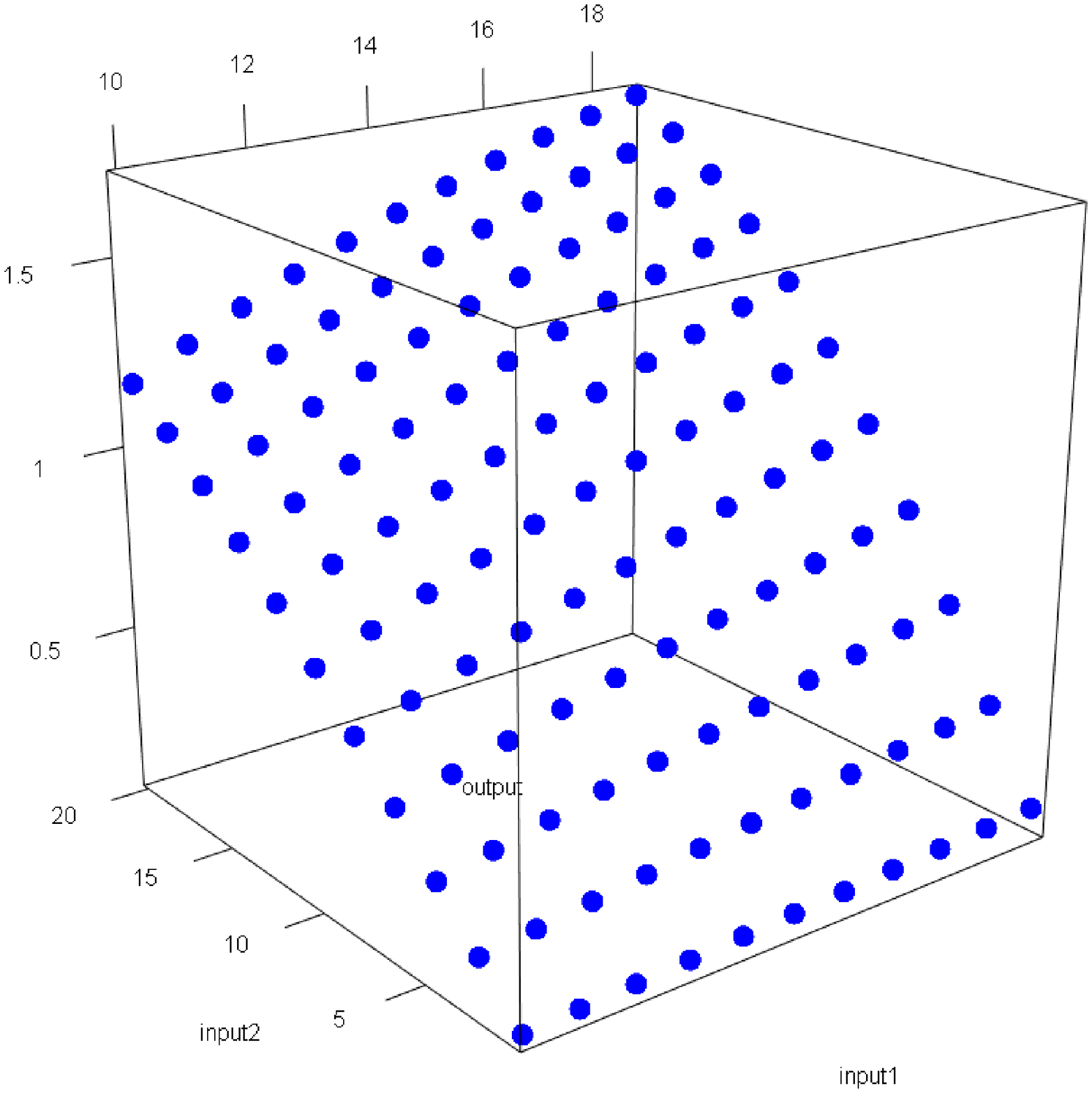}}\quad
  \subfloat[][\label{mode121}]{\includegraphics[width=.4\textwidth]{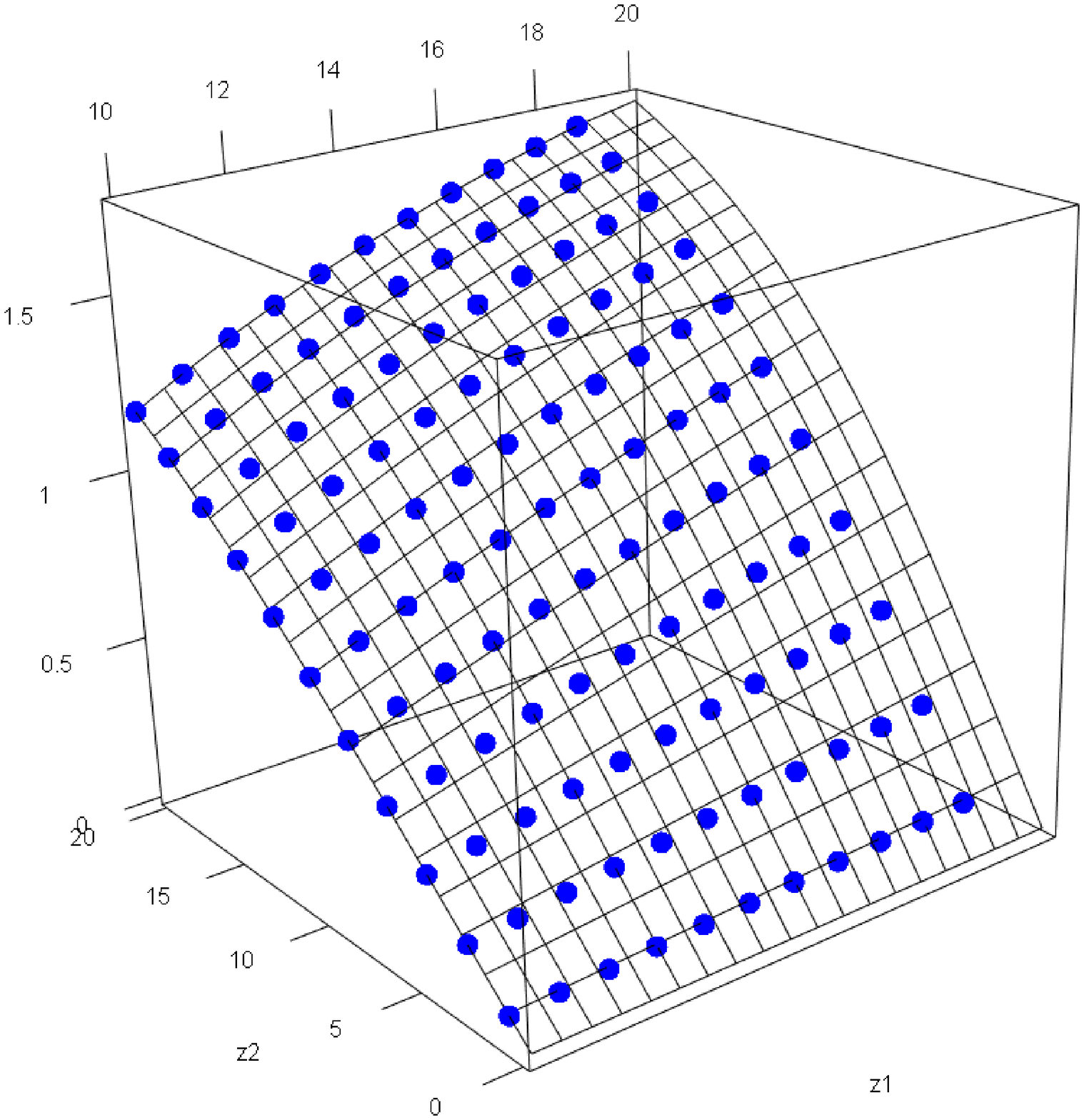}}\\
  \subfloat[][\label{5ptsLOO}]{\includegraphics[width=.4\textwidth]{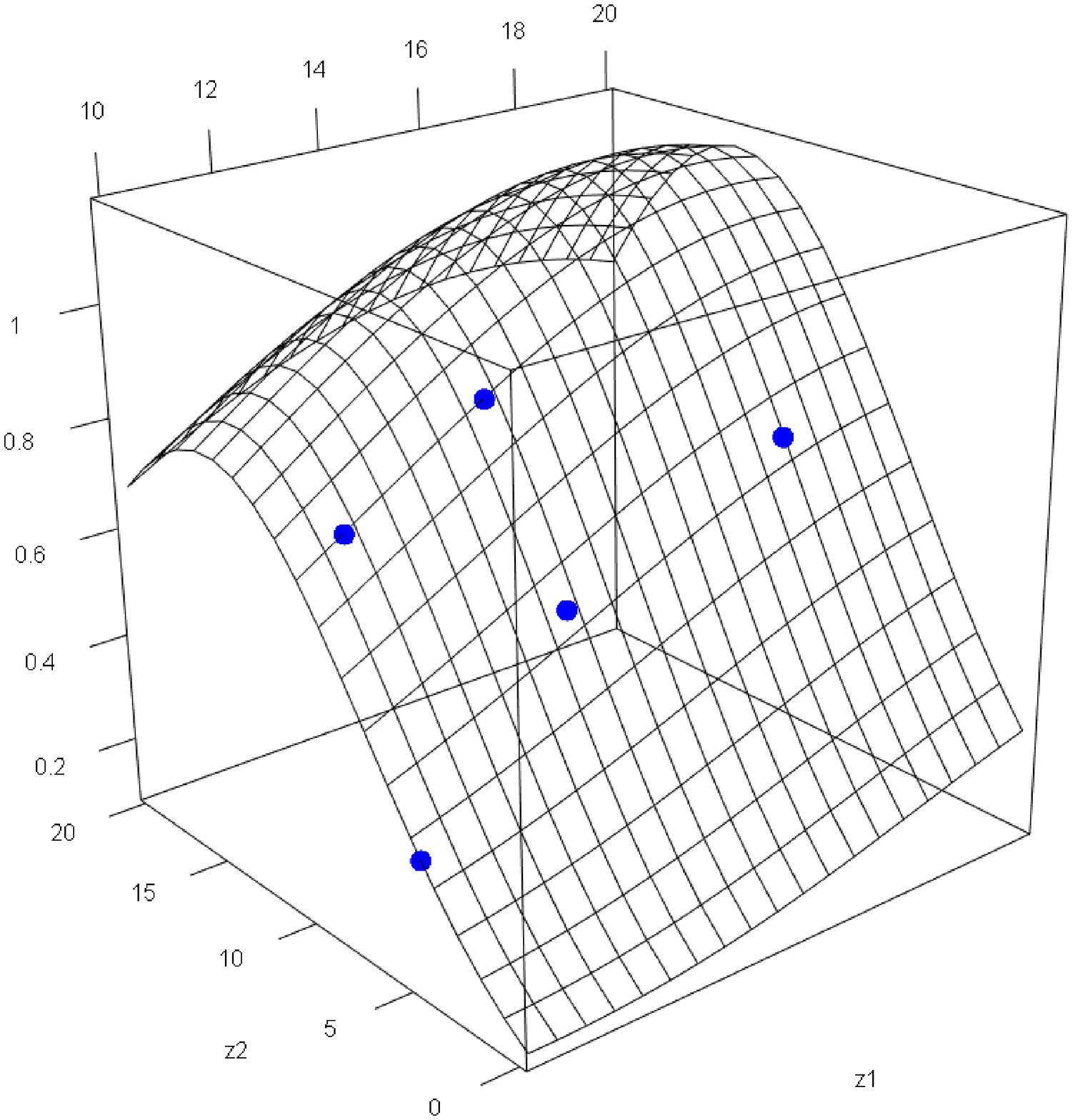}}\quad
  \subfloat[][\label{5ptsMode}]{\includegraphics[width=.4\textwidth]{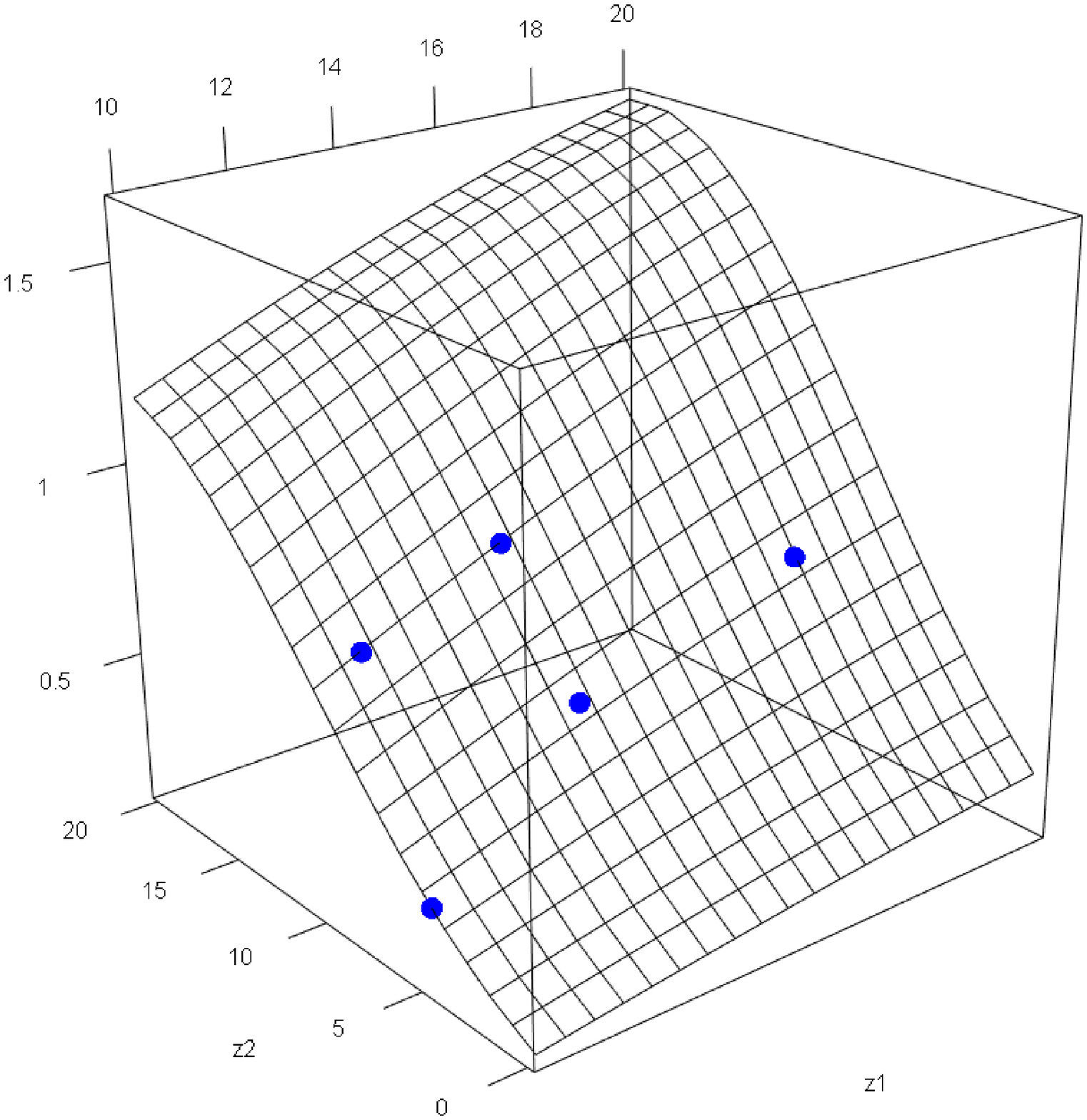}}
  \caption{(a) 3D visualization of  Godiva's data. (b) the maximum of the posterior distribution with 121 observations. (c) the unconstrained GP model with fives observations. (d) the constrained GP model with fives observations.}
  \label{monotonesurfaces}
\end{figure}
The one hundred and tweenty one observations defined on $[0,20]\times [10\times 20]$ (see Figure~\ref{godivaData}) have been used to show the efficiency of the proposed model in term of prediction and to compare it with the unconstrained Gaussian process model. The idea is to fix some of these observations and to test the quality of prediction of the proposed estimator at the other ones. The squared exponential covariance function has been used. In Figure~\ref{5ptsLOO}, we fix five observations using maximin Latin hypercube and we plot the unconstrained mean with the hyper parameters estimated by cross-validation methods. Notice that the unconstrained mean does not respect monotonicity constraints in the entire domain. In Figure~\ref{5ptsMode}, the same observations have been used to plot the maximum of the posterior distribution. The hyper parameters $(\theta_1,\theta_2)$ have been estimated by the suited cross-validation method \cite{Maatouk201538}.
We remark that with few observations, the proposed estimator verifies monotonicity constraints in the entire domain.  
Finally, the $Q^2$ criteria has been used to evaluate the quality of predictions
\begin{equation*}
Q^2=1-\frac{\sum_{i=1}^{n_t}(f(x_i)-\hat{f}(x_i))^2}{\sum_{i=1}^{n_t}(f(x_i)-\overline{y})^2},
\end{equation*}  
where $\hat{f}$ is the proposed estimator, $\overline{y}$ is the mean of the observations and $n_t$ is the number of tested data. The constrained model outperforms the unconstrained one with $Q^2$ equal to 0.98 versus 0.69.

\section{Conclusion}
In this paper, a finite-dimensional approximation of Gaussian processes to incorporate infinite number of inequality constraints (such as boundedness, monotonicity and convexity) and noisy observations is developed. It is based on a linear combination between Gaussian random coefficients and deterministic basis functions. The basis functions are chosen such that the infinite number of inequality constraints on the Gaussian process approximation are equivalent to a finite number of constraints on the coefficients. Consequently, simulate the conditional approximating process is equivalent to simulate a truncated Gaussian vector restricted to convex sets. By this methodology, the mean and the maximum of the posterior distribution are well defined. To show the performance of the proposed model in term of predictive accuracy and uncertainty quantification, a comparison with several recently models dealing with the same constraints is shown.

\section*{Acknowledgements}
Part of this work has been conducted within the frame of the ReDice Consortium, gathering industrial (CEA, EDF, IFPEN, IRSN, Renault) and academic (\'Ecole des Mines de Saint-\'Etienne, INRIA, and the University of Bern) partners around advanced methods for Computer Experiments. The author would like to thank Yann Richet (IRSN, Paris) for providing the nuclear safety data.

\bibliography{biblio}

\begin{thebibliography}{}

\bibitem[Abrahamsen and Benth, 2001]{Abrahamsen2001}
Abrahamsen, P. and Benth, F.~E. (2001).
\newblock {Kriging with inequality constraints}.
\newblock {\em Math. Geo.}, 33(6):719--744.

\bibitem[Aronszajn, 1950]{aronszajn1950}
Aronszajn, N. (1950).
\newblock Theory of reproducing kernels.
\newblock {\em Trans. Am. Math. Soc.}, 68:337--404.

\bibitem[Bay et~al., 2016]{bay2016}
Bay, X., Grammont, L., and Maatouk, H. (2016).
\newblock {Generalization of the Kimeldorf-Wahba correspondence for constrained
  interpolation}.
\newblock {\em Electron. J. Statist.}, 10(1):1580--1595.

\bibitem[Bay et~al., 2017]{Bay2017}
Bay, X., Grammont, L., and Maatouk, H. (2017).
\newblock {A new method for interpolating in a convex subset of a Hilbert
  space}.
\newblock {\em Comput. Optim. Appl.}, 68(1):95--120.
\newblock doi:10.1007/s10589-017-9906-9.

\bibitem[Berlinet and Thomas-Agnan, 2004]{berlinet2011reproducing}
Berlinet, A. and Thomas-Agnan, C. (2004).
\newblock {\em {Reproducing kernel Hilbert spaces in probability and
  statistics}}.
\newblock Kluwer Academic Publishers.

\bibitem[Botts, 2013]{Botts}
Botts, C. (2013).
\newblock {An accept-reject algorithm for the positive multivariate normal
  distribution}.
\newblock {\em Comput. Statist.}, 28(4):1749--1773.

\bibitem[Boyd and Vandenberghe, 2004]{Boyd2004CO993483}
Boyd, S. and Vandenberghe, L. (2004).
\newblock {\em {Convex optimization}}.
\newblock Cambridge University Press.

\bibitem[Chopin, 2011]{Chopin2011FST19607241960748}
Chopin, N. (2011).
\newblock {Fast simulation of truncated Gaussian distributions}.
\newblock {\em Stat. Comput.}, 21(2):275--288.

\bibitem[Cousin et~al., 2016]{Cousin2016}
Cousin, A., Maatouk, H., and Rulli\`ere, D. (2016).
\newblock Kriging of financial term-structures.
\newblock {\em Eur. J. Oper. Res.}, 255(2):631 -- 648.

\bibitem[Cramer and Leadbetter, 1967]{cramer1967stationary}
Cramer, H. and Leadbetter, R. (1967).
\newblock {\em Stationary and related stochastic processes: sample function
  properties and their applications}.
\newblock Wiley series in probability and mathematical statistics. Tracts on
  probability and statistics. Wiley.

\bibitem[Cressie and Johannesson, 2008]{cressie2008fixed}
Cressie, N. and Johannesson, G. (2008).
\newblock Fixed rank kriging for very large spatial data sets.
\newblock {\em J. R. Stat. Soc. B}, 70(1):209--226.

\bibitem[Delecroix et~al., 1996]{delecroix1996functional}
Delecroix, M., Simioni, M., and Thomas-Agnan, C. (1996).
\newblock Functional estimation under shape constraints.
\newblock {\em J. Nonparametr. Stat.}, 6(1):69--89.

\bibitem[Freulon and de~Fouquet, 1993]{Freulon1993}
Freulon, X. and de~Fouquet, C. (1993).
\newblock {Conditioning a Gaussian model with inequalities}.
\newblock In Soares, A., editor, {\em Geostatistics Tr{\'o}ia '92: Volume 1},
  pages 201--212. Springer Netherlands, Dordrecht.

\bibitem[Golchi et~al., 2015]{golchi2015monotone}
Golchi, S., Bingham, D., Chipman, H., and Campbell, D. (2015).
\newblock {Monotone emulation of computer experiments}.
\newblock {\em SIAM/ASA Journal on Uncertainty Quantification}, 3(1):370--392.

\bibitem[Goldfarb and Idnani, 1983]{Goldfarb83}
Goldfarb, D. and Idnani, A. (1983).
\newblock A numerically stable dual method for solving strictly convex
  quadratic programs.
\newblock {\em Math. Program.}, 27(1):1--33.

\bibitem[Holmes and Heard, 2003]{holmes2003generalized}
Holmes, C. and Heard, N. (2003).
\newblock Generalized monotonic regression using random change points.
\newblock {\em Stat. Med.}, 22(4):623--638.

\bibitem[Jidling et~al., 2017]{jidling2017linearly}
Jidling, C., Wahlstr{\"o}m, N., Wills, A., and Sch{\"o}n, T.~B. (2017).
\newblock {Linearly constrained Gaussian processes}.
\newblock {\em arXiv preprint arXiv:1703.00787}.

\bibitem[Kimeldorf and Wahba, 1970]{KW1970}
Kimeldorf, G.~S. and Wahba, G. (1970).
\newblock {A correspondence between Bayesian estimation on stochastic processes
  and smoothing by splines}.
\newblock {\em Ann. Math. Stat .}, 41(2):495--502.

\bibitem[Kleijnen and Van~Beers, 2013]{kleijnen2012monotonicity}
Kleijnen, J.~P. and Van~Beers, W.~C. (2013).
\newblock {Monotonicity-preserving bootstrapped kriging metamodels for
  expensive simulations}.
\newblock {\em J. Oper. Res. Soc.}, 64(5):708--717.

\bibitem[Lin and Dunson, 2014]{lin2014bayesian}
Lin, L. and Dunson, D.~B. (2014).
\newblock {Bayesian monotone regression using Gaussian process projection}.
\newblock {\em Biometrika}, 101(2):303--317.

\bibitem[Maatouk and Bay, 2016]{Maatouk2016}
Maatouk, H. and Bay, X. (2016).
\newblock {A new rejection sampling method for truncated multivariate Gaussian
  random variables restricted to convex sets}.
\newblock In Cools, R. and Nuyens, D., editors, {\em Monte Carlo and
  Quasi-Monte Carlo Methods}, pages 521--530. Springer International
  Publishing, Cham.

\bibitem[Maatouk et~al., 2015]{Maatouk201538}
Maatouk, H., Roustant, O., and Richet, Y. (2015).
\newblock {Cross-validation estimations of hyper-parameters of Gaussian
  processes with inequality constraints}.
\newblock {\em Procedia Environmental Sciences}, 27:38 -- 44.
\newblock Spatial Statistics conference 2015.

\bibitem[Neelon and Dunson, 2004]{neelon2004bayesian}
Neelon, B. and Dunson, D.~B. (2004).
\newblock Bayesian isotonic regression and trend analysis.
\newblock {\em Biometrics}, 60(2):398--406.

\bibitem[Parzen, 1962]{opac-b1081425}
Parzen, E. (1962).
\newblock {\em Stochastic processes}.
\newblock Holden-Day series in probability and statistics. Holden-Day, San
  Francisco, London, Amsterdam.

\bibitem[Philippe and Robert, 2003]{journals/sac/PhilippeR03}
Philippe, A. and Robert, C.~P. (2003).
\newblock Perfect simulation of positive {Gaussian} distributions.
\newblock {\em Stat. Comput.}, 13(2):179--186.

\bibitem[Ramsay, 1988]{ramsay1988}
Ramsay, J.~O. (1988).
\newblock Monotone regression splines in action.
\newblock {\em Statist. Sci.}, 3(4):425--441.

\bibitem[Ramsay, 1998]{ramsay1998estimating}
Ramsay, J.~O. (1998).
\newblock Estimating smooth monotone functions.
\newblock {\em J. R. Stat. Soc. B}, 60(2):365--375.

\bibitem[Rasmussen and Williams, 2006]{Rasmussen2005GPM1162254}
Rasmussen, C.~E. and Williams, C.~K. (2006).
\newblock {\em {Gaussian processes for machine learning}}.
\newblock MIT Press, Cambridge.

\bibitem[Riihim{\"a}ki and Vehtari, 2010]{journals/jmlr/RiihimakiV10}
Riihim{\"a}ki, J. and Vehtari, A. (2010).
\newblock Gaussian processes with monotonicity information.
\newblock {\em J. Mach. Learn. Res.}, 9:645--652.

\bibitem[Robert, 1995]{Robert}
Robert, C.~P. (1995).
\newblock Simulation of truncated normal variables.
\newblock {\em Stat. Comput.}, 5(2):121--125.

\bibitem[Saarela and Arjas, 2011]{saarela2011method}
Saarela, O. and Arjas, E. (2011).
\newblock {A method for Bayesian monotonic multiple regression}.
\newblock {\em Scand. J. Statist.}, 38(3):499--513.

\bibitem[Shively et~al., 2009]{shively2009bayesian}
Shively, T.~S., Sager, T.~W., and Walker, S.~G. (2009).
\newblock {A Bayesian approach to non-parametric monotone function estimation}.
\newblock {\em J. R. Stat. Soc. B}, 71(1):159--175.

\bibitem[Trecate et~al., 1999]{trecate1999finite}
Trecate, G.~F., Williams, C.~K., and Opper, M. (1999).
\newblock {Finite-dimensional approximation of Gaussian processes}.
\newblock In {\em Proceedings of the 1998 conference on Advances in neural
  information processing systems II}, pages 218--224. MIT Press.

\bibitem[Tutz and Leitenstorfer, 2007]{doi:10.1198/106186007X180949}
Tutz, G. and Leitenstorfer, F. (2007).
\newblock Generalized smooth monotonic regression in additive modeling.
\newblock {\em J. Comp. Graph. Stat.}, 16(1):165--188.

\bibitem[Wang and Berger, 2016]{wang2016estimating}
Wang, X. and Berger, J.~O. (2016).
\newblock {Estimating shape constrained functions using Gaussian processes}.
\newblock {\em SIAM/ASA Journal on Uncertainty Quantification}, 4(1):1--25.

\bibitem[Xuming and Peide, 1996]{He96monotoneb-spline}
Xuming, H. and Peide, S. (1996).
\newblock {Monotone B-spline smoothing}.
\newblock {\em J. Amer. Statist. Assoc.}, 93:643--650.

\end{thebibliography}
\bibliographystyle{apalike}

\end{document}